\documentclass[11pt]{llncs}

\usepackage{amssymb}
\usepackage{amsmath}
\usepackage{amsthm}
\usepackage{textcomp}
\usepackage{array}
\usepackage{multirow}

\usepackage{color}

\newcommand{\dollar}[0]{\$}

\newcommand{\LofC}[0]{\left\langle\right.}
\newcommand{\RofC}[0]{\left.\right\rangle}

\newtheorem{fact}{Fact}

\title{Tight bounds for the  space complexity of nonregular language recognition by real-time machines}
\author{Abuzer Yakary{\i}lmaz$ ^{\mbox{\tiny 1,}} $\thanks{Yakary{\i}lmaz was partially supported by the Scientific and Technological Research Council of Turkey (T\"{U}B\.ITAK) with grant 108E142 and FP7 FET-Open project QCS.}
\and
A.C. Cem Say$ ^{\mbox{\tiny 2,}} $\thanks{Say's work was partially supported by the Scientific and Technological
Research Council of Turkey (T\"{U}B\.ITAK) with grant 108E142.}}

\institute{$ ^{\mbox{\tiny 1}} $University of Latvia, Faculty of Computing, Raina bulv. 19, R\={\i}ga, LV-1586, Latvia \\
\email{abuzer@lu.lv} ~~ \\
$ ^{\mbox{\tiny 2}} $Bo\u{g}azi\c{c}i University, Department of Computer Engineering, Bebek 34342 \.{I}stanbul, Turkey
\\ \email{say@boun.edu.tr} ~~
 \\~~\\
\today
}

\begin{document}

\newlength{\twidth}
\maketitle
\pagenumbering{arabic}

\begin{abstract} \label{abstract:Abstract}

We examine the minimum amount of memory for real-time, as opposed to one-way, computation accepting nonregular languages. We consider deterministic, nondeterministic and alternating machines working within strong, middle and weak space, and processing general or unary inputs. In most cases, we are able to show that the lower bounds for one-way machines remain tight in the real-time case. Memory lower bounds for nonregular acceptance on other devices are also addressed. It is shown that increasing the number of stacks of real-time pushdown automata can result in exponential improvement in the total amount of space usage for nonregular language recognition.

\end{abstract}

\section{Introduction}
The effects of restricting a Turing machine (TM) to a single left-to-right pass of its input string have been well studied since the early years of computational complexity theory. An important distinction in this realm is the one between real-time computation \cite{Ra63B}, in which the machine is required to spend a single computational step on each input symbol, and  one-way computation \cite{SHL65}, where the input head is allowed to pause for some steps on the tape. These two modes are known to be equivalent in language recognition power in many setups. An important exception is quantum automata theory, where it has been shown that quantum versions of real-time finite automata \cite{AI99,YS11A} and pushdown automata \cite{YFSA12A} are strictly less powerful than their more general one-way versions in some modes of language recognition.

In this paper, we compare the real-time and one-way modes of computation from the point of view of the minimum amount of memory required by computers to recognize nonregular languages. 
A recent trend in the study of the lower limits of computation is the consideration of minimal amounts of combinations of various computational resources to perform certain tasks. Simultaneous lower bounds linking working space with input head reversals \cite{BMP94B,BMP94C,BMP95A,GMP98}, ambiguity degree \cite{BMP94A}, and runtime \cite{Pig09} have been established in the context of nonregular language recognition by many TM variants. Since real-time computation naturally involves the strictest possible time bound, our work can be seen as part of this literature.

Several variants of the TM model \cite{Si06} will be used in this study. All machines are equipped with a read-only input tape and a single worktape with a two-way read/write head. We will distinguish between the strong, middle, and weak modes of space complexity, whose definitions are reviewed below. We assume without loss of generality that none of the TMs we consider ever moves the work tape head to the left of its initial position.

\begin{itemize}
	\item A TM is said to be \textit{strongly $s(n)$ space-bounded} if all reachable configurations on any input of length $n$ have the work tape head no more than $s(n)$ cells away from its initial position.\footnote{We use the head position, rather than the number of nonblank cells \cite{Sz94}, for measuring space usage, since real-time quantum TMs are known \cite{YFSA12A} to be able to recognize nonregular languages by a technique that involves just moving the work head ($O(n)$ steps) to the right, without writing any nonblank symbols.}
	\item A TM is said to be \textit{middle $s(n)$ space-bounded} if all reachable configurations on any accepted input of length $n$ have the work tape head no more than $s(n)$ cells away from its initial position.
	\item A TM is said to be \textit{weakly $s(n)$ space-bounded} if, for any accepted input of length $n$, there exists at least one accepting computation containing no configurations that have the work tape head more than $s(n)$  cells away from its initial position.
\end{itemize}

In the rest of the paper, Section \ref{Sec:ATM} considers several specializations of the alternating TM model for both tally languages (i.e. those with unary input alphabets) and general ones, establishing that most minimum space bounds for nonregular language recognition coming from previous work on one-way versions of these machines remain tight in the real-time case. Similar questions about some other computational models like pushdown and counter automata are examined in Section \ref{Sec:diger}. Section \ref{Sec:conc} contains a list of open problems.

\section{Bounds for real-time alternating Turing machines}\label{Sec:ATM}

Table \ref{table:lower-bounds} depicts the known tight bounds for the minimum space complexity of nonregular language recognition by various specializations of the one-way alternating TM model for both general and unary input alphabets. (We refer the reader to \cite{Sz94,MP95} for the proofs of the facts in the table.) These lower bounds clearly carry over to real-time TMs. We will now proceed to prove that most of these bounds remain tight, by exhibiting real-time machines recognizing nonregular languages within the corresponding space bounds.

\begin{table}[h!]	
	\vskip\baselineskip
	\caption{Minimum space used by one-way TMs for recognizing nonregular languages. All bounds are tight.}
	\footnotesize
	\begin{center}
	\begin{tabular}{|l|l|l|l|l|l|l|}
	 	\hline
	 	& \multicolumn{3}{c|}{General input alphabet} & \multicolumn{3}{c|}{Unary input alphabet} 
	 	\\ \hline  	
	 	& Strong & Middle & Weak & Strong & Middle & Weak 
	 	\\ \hline
	 	Deterministic TM & $ \log n $ & $ \log n $ & $ \log n $ & $ \log n $ & $ \log n $ & $ \log n $
	 	\\ \hline
	 	Nondeterministic TM & $ \log n $ & $ \log n $ & $ \log \log n $ & $ \log n $ & $ \log n $ & $ \log \log n $
	 	\\ \hline
	 	Alternating TM & $ \log n $ & $ \log \log n $ & $ \log \log n $ & $ \log n $ & $ \log n $ & $ \log \log n $
	 	\\ \hline
	 \end{tabular}
	 \end{center}
	 \label{table:lower-bounds}
\end{table}~

\begin{theorem}
	\label{theorem:log}
	All the logarithmic bounds in Table \ref{table:lower-bounds} remain 
	tight for the real-time versions of the corresponding machines.
\end{theorem}
\begin{proof}
We construct a real-time deterministic TM $ \mathcal{D} $ that recognizes the nonregular tally language 
\[
	L_{\mathcal{D}} = \{ a^{k_i} \mid i \geq 0, k_0 = 8, \mbox{and } k_{i+1} = k_i + 2^i(i+1)+2 \mbox{ for } i \geq 0 \},
\]
using only logarithmic space.

 The input string is assumed to be followed by the endmarker $\dollar$ on the tape. The  work tape symbols are $ \LofC $, $ \RofC $, $0$,  $1$, and the blank symbol $\#$. $ \mathcal{D} $ will maintain a counter in reverse binary representation between the delimiters  $ \LofC $ and $ \RofC $ on its work tape. At the beginning of the computation, $ \mathcal{D} $ expects to read four  $ a $'s from the input, (rejecting otherwise,) and moves the work tape head right and then left during these four steps to write the string $ \LofC 0 \RofC $, and position the head on the symbol $ \LofC $. Having initialized its counter, $ \mathcal{D} $ enters the following loop:
\begin{itemize}
	\item the work tape head goes from symbol $ \LofC $ to symbol $ \RofC $ 
		while incrementing the counter by one, and then,
	\item the work tape head goes from symbol $ \RofC $ to symbol $ \LofC $, while checking 
		whether the counter value is a power of two or not.
\end{itemize}
Notice that the work tape head never remains stationary on any symbol.		
Moreover, if required, the symbol $ \RofC $ on the work tape is shifted 
one square to the right.
If $ \mathcal{D} $ reads $ \dollar $ whenever the work tape head is not placed on the symbol $ \LofC $,
the input is rejected.
When the work tape head is placed on symbol $ \LofC $ and 
the currently scanned input symbol is $ \dollar $,
then $ \mathcal{D} $ accepts if the counter is a power of two, and rejects otherwise.

The shortest member of $L_{\mathcal{D}}$ is $ a^{8} $, since the work tape of $ \mathcal{D} $ contains $ \LofC 1 \RofC \#\#\cdots $,  and the work tape head is placed on the $ \LofC $ after reading the input $ a^{8} $.

For any $i>0$, suppose that  the work tape contains $ \LofC 0^{i-1} 1 \RofC \#\#\cdots $ with the head scanning the $ \LofC $ after reading the string $ a^{k_{i-1}} $ from the input. If the input head scans a $\dollar$ at this point, $ \mathcal{D} $ will accept. If the input is longer, the next opportunity for acceptance will come when $ \LofC 0^{i} 1 \RofC \#\#\cdots $ is written on the work tape, and the head is back on the $ \LofC $. To reach that configuration, the counter will have to be incremented $  2^i - 2^{i-1} =  2^{i-1}$ times. Each such increment involves the work tape head going all the way to the $\RofC $ and then back to the $ \LofC $.  Each of these round trips except the last one takes $ (i+1) + (i+1) $ steps, whereas the last one takes  $ (i+2) + (i+2) $ steps, since it involves extending the length of the counter by 1. The total number of steps between the two consecutive potentially accepting configurations is therefore $ (2^{i-1}-1) (2i+2) + 2i+4 = 2^i(i+1)+2 $.



The fact that  $  k_{i}  -  k_{i-1}  $ increases as $ i $ increases leads to a trivial proof, using the pumping lemma for regular languages, of the nonregularity of $ L_{\mathcal{D}} $.

Since $ \mathcal{D} $ belongs to the most specialized machine class under consideration, and obviously uses $O(\log n)$ space to recognize a unary nonregular language, we conclude that all the logarithmic lower bounds in Table \ref{table:lower-bounds} are tight for real-time machines as well.
\end{proof}

We note that Freivalds and Karpinski use a real-time logarithmic-size counter in order to demonstrate a lower time bound for probabilistic TMs in \cite{FK95}. Such a counter could also be used to prove at least some of the bounds handled by Theorem \ref{theorem:log} to be tight. To the best of our knowledge, the minimum space requirements for nonregular language recognition by these real-time machines have not been studied before.

We will now turn our attention to the double logarithmic bounds for machines with general input alphabets in Table \ref{table:lower-bounds}. Our strategy is to modify the one-way machines which are used in demonstrating these results to obtain real-time machines with the same space complexity. The languages recognized by these new real-time machines will be ``padded" versions of their counterparts in the one-way setup. Essentially, the real-time machine will consume a padding symbol ($ \kappa $) for each step of the one-way machine in which the input head pauses on the tape.

\begin{definition}	
	Given a one-way deterministic, nondeterministic or alternating TM $ \mathcal{D} $,
the real-time TM $ \mathcal{D}_{\kappa} $ of the same model variant is constructed as follows:
	\begin{itemize}
		\item The input alphabet of $ \mathcal{D}_{\kappa} $ is $ \Sigma_{\kappa} = \Sigma \cup \{ \kappa \} $, where $ \Sigma $ is the input alphabet of $ \mathcal{D} $, and $ \kappa \notin \Sigma$; 
	\item The work tape alphabets of the two machines are identical;
	\item $ \mathcal{D}_{\kappa} $ skips over any prefix of $\kappa$'s in its input until it sees the first  non-$ \kappa $
			symbol, at which point it starts its emulation of $ \mathcal{D} $;
		\item On each input symbol other than $ \kappa $,
			$ \mathcal{D}_{\kappa} $ emulates the state transition and work tape action of $ \mathcal{D} $ for that symbol;
		\item After implementing a stationary transition (that is, one which does not move the input head,) of $ \mathcal{D} $ in the manner described above, $ \mathcal{D}_{\kappa} $ checks whether the next input symbol is  $ \kappa $. If this is the case, $ \mathcal{D}_{\kappa} $ implements the transition that $ \mathcal{D} $ would be making while scanning the symbol that
			is the last non-$ \kappa $ symbol scanned by $ \mathcal{D}_{\kappa} $,
			otherwise,
			$ \mathcal{D}_{\kappa} $ rejects the input;
		\item After implementing a transition of $ \mathcal{D} $ that moves its input head to the right, $ \mathcal{D}_{\kappa} $ skips over any $ \kappa $'s in the input without any state or work tape change, until it sees a non-$ \kappa $
			symbol.
	\end{itemize}
	At the end, $ \mathcal{D}_{\kappa} $ accepts or rejects as $ \mathcal{D} $ would, unless it has already rejected in the eventualities described above.
\end{definition}

Let $ h_{\kappa} : \Sigma_\kappa \rightarrow \Sigma $ be a homomorphism such that 
\begin{itemize}
	\item $ h_{\kappa}(x)=x $ if $ x \neq \kappa $, and,
	\item $ h_{\kappa}(\kappa)= \varepsilon $.
\end{itemize}

\begin{definition}
	$ L_{\mbox{-}\kappa} \subset \Sigma_{\kappa}^{*} $ is the language recognized by $ \mathcal{D}_{\kappa} $, where $ \mathcal{D} $ is a one-way TM recognizing the language $ L \subseteq \Sigma^{*} $.
\end{definition}
\begin{lemma}
	$ h_{\kappa} ( L_{\mbox{-}\kappa} ) = L $.
\end{lemma}
\begin{proof}	
	It is obvious that $ h_{\kappa} ( L_{\mbox{-}\kappa} ) \subseteq L $. Assume that a string  $ w \in \Sigma^{*} $ is accepted by the one-way alternating TM $ \mathcal{D} $. $ \mathcal{D} $ may spend different amounts of time, pausing on different symbols, in different branches of its accepting computation tree. Let $t$ be the maximum number of steps that $ \mathcal{D} $ pauses for on any particular symbol in this tree. The string $w'$, obtained by inserting $t$ $\kappa$'s after every symbol of $w$, is accepted by  $ \mathcal{D}_{\kappa} $.
\end{proof}

\begin{remark}
	If $ L $ is not a member of class $ \mathcal{C} $ that is closed under homomorphism,
	then $ L_{\mbox{-}\kappa} $ is also not a member of $ \mathcal{C} $. In particular, if $L$ is nonregular, then $ L_{\mbox{-}\kappa} $ is also nonregular.
\end{remark}

\begin{theorem}\label{theorem:strong}
	If $ L $ is recognized by a  strongly (resp., middle) space $ s(n) $-bounded one-way TM  $ \mathcal{D} $, then for any nondecreasing function $ t(n) $ such that $ s(n) \in O(t(n)) $, the real-time TM $ \mathcal{D}_{\kappa} $ recognizing
	$ L_{\mbox{-}\kappa} $ is strongly  (resp., middle) $ O(t(n)) $-space bounded.
\end{theorem}
\begin{proof}	
	For any $ w \in \Sigma^{*} $, the space used by $ \mathcal{D}_{\kappa} $ on any input which is a 
	member of $ \{ w_{\kappa} \in \Sigma_{\kappa}^{*} \mid h_{\kappa} ( w_{\kappa} ) = w \} $
	is at most equal to the space used by $ \mathcal{D} $ on $ w $.
\end{proof}


Consider the nonregular language 

$ L_{gcm} = \{ a^{m} b^{M} \mid M \mbox{ is a common multiple of all } i \leq m \} $.

Szepietowski \cite{Sz88} showed that $ L_{gcm} $ is recognized by a middle space $ \log \log (n) $-bounded one-way alternating TM that we will call $ \mathcal{A} $.

\begin{corollary}
    The real-time alternating TM $ \mathcal{A}_{\kappa} $ recognizing the nonregular language
	$ L_{gcm \mbox{-}\kappa} $ is middle $ O(\log\log(n)) $-space bounded.
\end{corollary}

We are therefore able to state that the middle and strong space bounds for alternating TMs with general input alphabets in Table \ref{table:lower-bounds} are also tight for the real-time versions of these machines.

\begin{remark}\label{remark:weak}
There is no analogue of Theorem \ref{theorem:strong} 
 for weak space, because some accepting path may use more space but use less time (i.e. perform fewer stationary moves) than the accepting path with the best space usage in the one-way machine, and the padding process described above may then yield some strings whose only accepting paths use unacceptably large amounts of  space in the real-time machine. The cases covered in Theorem \ref{theorem:strong} avoid this problem by requiring all computations in consideration to remain within the space bound.
\end{remark}

To examine the weak space bound case for nondeterministic machines with general input alphabets, we focus on the one-way machine $ \mathcal{N} $, which recognizes $ L_{j \neq k} = \{ a^{j}b^{k} \mid j \neq k \} $, \cite{Sz94} (Lemma 4.1.3 on Page 23) used to establish the corresponding lower bound in Table \ref{table:lower-bounds}:
The work tape of $ \mathcal{N} $ is divided into three tracks. At the beginning of the computation, $ \mathcal{N} $ nondeterministically selects a number $ l > 1 $, and writes it in binary on the top track. Any string containing an $a$ after a $b$ is rejected. For strings of the form $a^{r}b^{s}$, $ \mathcal{N} $ calculates and stores the value of $ r \mod(l) $ (resp., $ s \mod(l) $) on 
the middle track (resp., bottom track).
When the end-marker is read, the numbers in the middle and bottom tracks are compared.
If they are not equal, the input is accepted.
The nondeterministic path corresponding to $ l $, namely $ \mathsf{npath}_{l} $,
needs only $ \Theta(\log(l)) $ space.

If $r = s$, then $ r \equiv s \mod(l) $ for all $ l > 1 $, and the input would be rejected. Using a number-theoretical fact (\cite{Sz94}, page 22) which states that, whenever $ r \neq s$, there exists a number  $ l \in O(\log(r+s)) $ such that $ r \not\equiv s \mod(l) $, one concludes that $ \mathcal{N} $ is weakly $ \log\log(n) $-space bounded, and recognizes the nonregular language $ L_{j \neq k} $. 

We will modify $ \mathcal{N} $ to obtain a real-time TM recognizing a padded version of $ L_{j \neq k}$, taking extra care to handle the additional complications of weak space bounds mentioned in the above remark. We do this by making sure that nondeterministic paths with larger values of $l$ have greater runtimes than those with small values of $l$.

We rewrite the program of $ \mathcal{N} $ such that $ d_{l} $, that is, the number of stationary steps performed during the processing of a single input symbol by the nondeterministic path containing $l$ on the top track, depends solely on $l$.
For example, we can use the following strategy:
\begin{itemize}
	\item The work tape head is always placed on the leftmost nonblank symbol before reading the next input symbol, and,
	\item The operations on the work tape are performed by sweeping the work tape head,
		which operates at the speed of one cell per step,
		between the leftmost and rightmost nonblank symbols $ c > 0 $ times.
\end{itemize}
By selecting a sufficiently large value for $ c $,
$ d_{l} $ can be set to $ c \lceil \log l \rceil + k $ for nondeterministic path $ \mathsf{npath}_{l} $,
where the exact value of $ k \in \mathbb{Z} $ depends on the way the numbers are stored on the work tape.
Thus, we can guarantee that for any $ l' > l $,
the number of the stationary steps (on the input tape) 
of $ \mathsf{npath}_{l'} $ cannot be less than that of $ \mathsf{npath}_{l} $.
Let $ \mathcal{N}' $ be the one-way nondeterministic TM implementing this modified algorithm,
and $ L^{'}_{(j \neq k)\mbox{-}\kappa} $ be the language recognized by $ \mathcal{N}_{\kappa}' $. 
\begin{lemma}
 The real-time nondeterministic TM $ \mathcal{N}_{\kappa}' $ recognizing the nonregular language
	$ L^{'}_{(j \neq k)\mbox{-}\kappa} $  is weakly $ O(\log\log(n)) $-space bounded.
\end{lemma}
\begin{proof}
	Let $ w = a^{r}b^{s} $ be a member of $ L_{j \neq k} $, and $ l $ be the smallest
	number satisfying  $ r \not\equiv s \mod (l) $. 
	Then, for all $ w_{\kappa} \in L^{'}_{(j \neq k) \mbox{-}\kappa} $ satisfying $ h_{\kappa} (w_{\kappa}) = w $,
	the nondeterministic path of $ \mathcal{N}_{\kappa}' $ corresponding to $ l $  accepts the input.
\end{proof}

Table \ref{table:tight-bounds} summarizes the results presented in this section.

\begin{table}[h!]	
	\vskip\baselineskip
	\caption{Lower space bounds for nonregular language recognition by real-time TMs. The bounds we have shown to be tight are in boldface.}
	\footnotesize
	\begin{center}
	\begin{tabular}{|l|l|l|l|l|l|l|}
	 	\hline
	 	& \multicolumn{3}{c|}{General input alphabet} & \multicolumn{3}{c|}{Unary input alphabet} 
	 	\\ \hline  	
	 	& Strong & Middle & Weak & Strong & Middle & Weak 
	 	\\ \hline
	 	Deterministic TM & $ \mathbf{log} \mspace{2mu} n $ & $ \mathbf{log} \mspace{2mu} n $ &
	 		$ \mathbf{log} \mspace{2mu} n $ & $ \mathbf{log} \mspace{2mu} n $ &
	 		$ \mathbf{log} \mspace{2mu} n $ & $ \mathbf{log} \mspace{2mu} n $
	 	\\ \hline
	 	Nondeterministic TM & $ \mathbf{log} \mspace{2mu} n $ & $ \mathbf{log} \mspace{2mu} n $ &
	 		$ \mathbf{log} \mspace{2mu} \mathbf{log} \mspace{2mu} n $ & $ \mathbf{log} \mspace{2mu} n $ &
	 		$ \mathbf{log} \mspace{2mu} n $ & $ \log \log n $
	 	\\ \hline
	 	Alternating TM & $ \mathbf{log} \mspace{2mu} n $ & $ \mathbf{log} \mspace{2mu} \mathbf{log} \mspace{2mu} n $ &
	 		$ \mathbf{log} \mspace{2mu} \mathbf{log} \mspace{2mu} n $ & $ \mathbf{log} \mspace{2mu} n $ &
	 		$ \mathbf{log} \mspace{2mu} n $ & $ \log \log n $
	 	\\ \hline 
	 \end{tabular}
	 \end{center}
	 \label{table:tight-bounds}
\end{table}

\section{Other real-time computational models}\label{Sec:diger}

The question of the minimum amount of useful space is also interesting for more restricted models, like pushdown or counter automata. The following facts are known about (one-way) pushdown automata (PDAs) with unrestricted pushdown alphabets: 

\begin{fact}\label{fact:Gabarro}
 All deterministic PDAs which recognize nonregular languages are weakly $\Theta(n)$-space bounded \cite{Ga84}.
\end{fact}
\begin{fact}
	There exists a weakly $\log n$-space bounded nondeterministic PDA that recognizes a nonregular language \cite{Re07}.
\end{fact}

Counter automata \cite{FMR67} (essentially, PDAs with a unary pushdown alphabet,) do not seem to have been investigated deeply. As an easy corollary of Fact \ref{fact:Gabarro}, we have
\begin{fact}
	$o(n)$-space bounded real-time deterministic counter automata cannot recognize nonregular languages.
\end{fact}


We have considered only a single work tape in our real-time models until now. It is known \cite{FMR67} that increasing the number of work tapes does increase the computational power of these machines. We show that deterministic PDAs and counter automata with two work tapes can recognize nonregular languages using sublinear space, doing better that the bound of Fact \ref{fact:Gabarro} for single-stack machines:

Let $ (i)_{2} $  denote the binary representation of $ i \in \mathbb{N} $. Let $ (i)^{r}_{2} $ denote the reverse of $ (i)_{2} $.
Consider the language
\begin{equation}\label{equation:erb}
	 L_{even-rev-bins} = 
	 \{ a(0)_{2}a (1)_{2}^{r}a(2)_{2}a (3)_{2}^{r}a \cdots a (2k)_{2} a (2k+1)_{2}^{r} \mid k>0 \}.
\end{equation}
In the following, substrings of the form $\{0,1\}^{*}$ delimited by a's will be called ``blocks," and the term $block_{i}$ will denote the $i$'th block to be encountered by the one-way input head. 
\begin{theorem}\label{theorem:2pda}
	$ L_{even-rev-bins} $ can be recognized by a deterministic real-time PDA with two stacks, and the total amount of space used on the stacks for accepted strings is $ O(\log n) $.  
\end{theorem}
\begin{proof}
The machine uses one stack for checking that the members of the pairs in 
$\{(block_{1},block_{2}), (block_{3},block_{4}),...\}$ are related according to the language definition, whereas the other stack is used for the set $\{(block_{2},block_{3}),(block_{4},block_{5}),...\}$. 

Inputs that cause the machine to compare blocks all the way to the last one, say, $block_{i}$, incur the greatest space cost. The space used in the stacks in that case is $O(\log i)$, and the length of the input prefix is clearly more than $i$. We conclude that the machine uses $O(\log n)$ space.  
\end{proof}

\begin{theorem}\label{theorem:root}
For any $j > 1$, there exists a nonregular language that can be recognized by a deterministic real-time automaton with $j$ (unary) counters, such that the total amount of space used on the counters is $ O(n^{\frac{1}{j}}) $ for all input strings. 

\end{theorem}
\begin{proof}
We start by considering the language
\[ L_{\textbf{2}} = \{a_{1}a_{0}a_{1}a_{0}^{2}a_{1}a_{0}^{3}a_{1} \cdots a_{1} a_{0}^{k} \mid k \in \mathbb{N} \}, \]
(inspired by \cite{ABB80,Ga84},) on the alphabet $\{a_{0},a_{1}\}$. $L_{\textbf{2}} $ is recognized by a deterministic real-time automaton with two  counters.	
	The algorithm, say, $\mathcal{A}_{2}$, is similar to, in fact simpler than, that  of Theorem \ref{theorem:2pda}. As for the space usage, assume that the machine arrives at a decision to accept or reject after scanning the $i$'th block of $a_{0}$'s. Since the length of the scanned input prefix is more than $ \sum_{l=1}^{i-1}l+1 = \Theta(i^{2})$, and the counters have used $O(i)$ space up to that point, the machine is strongly $O(\sqrt{n})$-space bounded. Let $w_{\textbf{j},k}$ denote the $k$'th shortest element of $ L_{\textbf{j}} $. Note that $\mathcal{A}_{2}$ accepts $w_{\textbf{2},k}$ with one counter containing zero, and the other counter containing $k$.
	
	We define 	
	\[ L_{\textbf{3}} = \{a_{2} w_{\textbf{2},1}a_{2} w_{\textbf{2},2}a_{2} w_{\textbf{2},3} a_{2} \cdots a_{2} w_{\textbf{2},k}  \mid k \in \mathbb{N} \} \]	
on the alphabet $\{a_{0},a_{1},a_{2}\}$. 	A deterministic real-time machine with three counters can recognize $ L_{\textbf{3}}$ as follows. The automaton processes each block delimited by $a_{2}$'s by using two of the counters for implementing $\mathcal{A}_{2}$. If  the input head arrives at the $i$'th $a_{2}$ ($i>1$), two of the counters contain zero, whereas the third counter contains $i-1$. The two counters containing zero are now used to run $\mathcal{A}_{2}$ on the $i$'th block, while the remaining counter is used to check that the number of $a_{1}$'s in the $i$'th block is indeed $i$. Two counters have value zero, and the third one contains $k$ when this machine accepts $w_{\textbf{3},k}$. By a reasoning similar to the one about $ L_{\textbf{2}}$ above, the space usage is $ O(n^{\frac{1}{3}}) $. The generalization to  
		\[ L_{\textbf{j+1}} = \{a_{j} w_{\textbf{j},1}a_{j} w_{\textbf{j},2}a_{j} w_{\textbf{j},3} a_{j} \cdots a_{j} w_{\textbf{j},k} \mid k \in \mathbb{N} \}	\]
	is now simple.
	
\end{proof}

Even with one work tape, the probabilistic versions of these machines \cite{HS10,Fr79} also require less space than the deterministic versions for recognizing some languages. The language $L_{even-rev-bins}$ of Equation \ref{equation:erb} can be recognized by a real-time probabilistic PDA with error bound $\frac{1}{3}$ as follows: Reject the input at the start with probability $\frac{1}{3}$. With the remaining probability, toss a fair coin to split to two computational paths, each using its single stack to mimic the corresponding stack of Theorem \ref{theorem:2pda}. If the input is in $L_{even-rev-bins}$, both paths accept, yielding an overall acceptance probability of $\frac{2}{3}$. Otherwise, at least one path rejects, resulting in an acceptance probability of at most $\frac{1}{3}$. The space requirement for members of the language is $O(\log n)$, as in Theorem \ref{theorem:2pda}.

With a similar approach, real-time probabilistic automata with one counter that can recognize any language of the $ L_{\textbf{j}}$  family from Theorem \ref{theorem:root} using $ O(n^{\frac{1}{j}}) $ space for member strings, albeit with error bounds that increase with $j$, can be constructed. 

\section{Open questions}\label{Sec:conc}

We conclude with a list of problems for future work.

An obvious question left open in Section \ref{Sec:ATM} is whether the double logarithmic lower bounds for the recognition of nonregular tally languages by real-time nondeterministic and alternating TMs are tight. 

How does alternation affect the amount of useful space for (both one-way and real-time) PDAs and counter automata?

What are the lower total space bounds for nonregular language recognition for machines with multiple work tapes, as exemplified in Section \ref{Sec:diger}?

Real-time ``nondeterministic" quantum Turing machines (i.e., those that accept with nonzero probability if and only if the input is a member of the language) are known \cite{YS10A} to be able to recognize nonregular languages even with constant space, whereas no such language can be recognized  by small-space probabilistic Turing machines in this mode, since these are synonymous with nondeterministic TMs (Table \ref{table:tight-bounds}). When two-sided unbounded error is allowed, probabilistic finite automata can recognize nonregular languages \cite{Ra63A}. What are the minimum space requirements for real-time nonregular language recognition  of probabilistic TMs in the   bounded-error regime \cite{Fr85}? Can one improve on the space usage of the probabilistic machines mentioned at the end of Section \ref{Sec:diger}?

What is the minimum amount of space required for a real-time
quantum Turing machine to recognize a nonregular language with bounded
error?



\textbf{Acknowledgements.} We are grateful for the constructive comments of the two anonymous reviewers. We also thank Stefan D. Bruda, Klaus Reinhardt, Giovanni Pighizzini, Juraj Hromkovi\v{c}, and R\={u}si\c{n}\v{s} Freivalds  for their helpful answers to our questions. 

\bibliographystyle{alpha}
\bibliography{YakaryilmazSay}

\end{document}